\documentclass[12pt,a4]{article}
\pdfoutput=1
\usepackage{amsmath, amsthm, amsfonts}
\usepackage{graphicx}
\usepackage[affil-it, auth-sc-lg]{authblk}
\graphicspath{{Plots/}}
\usepackage{natbib}
\allowdisplaybreaks[1]
\newtheorem{thm}{Theorem}

\DeclareSymbolFont{AMSb}{U}{msb}{m}{n}
\DeclareMathSymbol{\A}{\mathord}{AMSb}{"41}
\DeclareMathSymbol{\B}{\mathord}{AMSb}{"42}
\DeclareMathSymbol{\Fban}{\mathord}{AMSb}{"46}
\DeclareMathSymbol{\Gban}{\mathord}{AMSb}{"47}
\DeclareMathSymbol{\R}{\mathord}{AMSb}{"52}
\DeclareMathSymbol{\bS}{\mathord}{AMSb}{"53}
\DeclareMathSymbol{\Hil}{\mathord}{AMSb}{"48}
\DeclareMathSymbol{\Kil}{\mathord}{AMSb}{'113}
\newcommand{\1}{\mathord{\boldsymbol{1}}}
\newcommand{\Lin}{\mathord{\mathcal{L}}}
\newcommand{\cF}{\mathord{\mathcal{F}}}
\newcommand{\cG}{\mathord{\mathcal{G}}}
\newcommand{\cB}{\mathord{\mathcal{B}}}
\renewcommand{\P}{\mathord{\mathrm{P}}}
\newcommand{\E}{\mathord{\mathrm{E}}}
\newcommand{\Cov}{\mathord{\mathrm{Cov}}}
\newcommand{\I}{\mathord{\mathrm{I}}}
\newcommand{\Nor}{\mathord{\mathcal{N}}}

\def\keywords{\vspace{.5em}
{\textit{Keywords}:\,\relax%
}}
\def\endkeywords{\par}

\begin{document}

\author{Giovanni Petris%
\thanks{Email: \texttt{gpetris@gmail.com}}}
\affil{University of Arkansas}
\title{A Bayesian framework for functional\\ time series analysis} 
\date{November 2013}

\maketitle

\begin{abstract}
  The paper introduces a general framework for statistical analysis of
  functional time series from a Bayesian perspective. The proposed
  approach, based on an extension of the popular dynamic linear model
  to Banach-space valued observations and states, is very flexible but
  also easy to implement in many cases. For many kinds of data, such
  as continuous functions, we show how the general theory of
  stochastic processes provides a convenient tool to specify priors
  and transition probabilities of the model. Finally, we show how
  standard Markov chain Monte Carlo methods for posterior simulation
  can be employed under consistent discretizations of the data.  
\end{abstract}

\section{Introduction} 
Time series data consisting of individual high or infinite dimensional
observations are becoming more and more common in many applied
areas. As a consequence there is a need to develop models and
algorithms for the analysis and forecasting of this kind of
data. Clearly any statistically sound model should account for the
temporal dependence of the data, in addition to a possibly complex
correlation structure within the observations made at a specific time
point. Statistician have been working on methods for the analysis of
functional data for several years; early references are
\citet{Ramsay+Dalzell:1991} and \citet{Grenander:1981}, the books by
\citet{Ramsay+Silverman:2006, Ramsay+Silverman:2002},
\citet{Ramsay+Hooker+Graves:2009}, and \citet{Ferraty+Vieu:2006}
provide good entry points to the recent literature. More recently,
there has been an interest also in models and tools for time series of
functional data, see for example
\citet{Aue+Gabrys+Horvath+Kokoszka:2009, Ding+Lin+Zhong:2009,
  Horvath+Huskova+Kokoszka:2010, Hyndman+Shang:2009,
  Kargin+Onatski:2008, Shen:2009}.  The review papers by
\citet{Mas+Pumo:2010} and \citet{Hormann+Kokoszka:2012} contain
up-to-date references, while the books by \citet{Bosq:2000} and
\citet{Horvath+Kokoszka:2012} provide a comprehensive
background. However, a fatisfactory treatment of functional time
series from a Bayesian perspective has been so far elusive. Our
research aims at filling this gap, providing a flexible and
easy-to-use class of models for Bayesian analysis of functional time
series.

From a methodological point of view, the main focus of the paper is
the extension of the highly successful dynamic linear model to
function spaces. Related references are \citet{Falb:1967} and
\citet{Bensoussan:2003}. Their extensions, however, are different from
the one suggested in the present paper, since they focus on
continuous-time processes and, most importantly, they do not provide
algorithms that are well-suited for practical applications. 

The layout of the paper is as follows. We introduce in
Section~\ref{sec:rv} the basic notions related to Banach space-valued
random variables that will be needed for the subsequent
development. The model we propose is discussed in
Section~\ref{sec:fdlm}, where also Kalman filter and smother in the
infinite dimensional setting are discussed. Section~\ref{sec:example}
focus on a practical example in which the model is applied to a
time series of continuous functions. Concluding remarks are
contained in Section~\ref{sec:conclusion}.

\section{Functional random variables}
\label{sec:rv} %
In this section we briefly introduce Banach space-valued random
variables and the extension of the notions of expectation and
covariance to this type of random variables. We also discuss Gaussian
distributions on Banach spaces. Most of the material of this section
is covered in great detail in the monographs by \citet{Bogachev:1998}
and \citet{DaPrato+Zabczyk:1992}.  We consider in the following only
separable Banach spaces. While this is a strong requirement from a
theoretical perspective, it is not a serious limitation for
applications, since almost all Banach spaces of functions used in
practice are separable, the most notable exception being probably
$L^\infty$, the space of essentially bounded (equivalence classes of)
functions on a given measurable space. The symbol $\B$, possibly with a
subscript, will be used to denote a separable, but otherwise general,
Banach space. We will use the notation $\Lin(\B, \B_1)$ for the Banach
space consisting of all continuos linear operators mapping $\B$ to
$\B_1$. In the case when $\B_1 = \B$ we will abbreviate the notation
to $\Lin(\B)$. $\B^*$ will denote the Banach space of all continuous
linear functionals on $\B$, i.e., $\B^* = \Lin(\B, \R)$.  Recall that,
for any $A\in \Lin(\B, \B_1)$, the adjoint operator $A^*\in
\Lin(\B_1^*, \B^*)$ is defined, for every $b_1^*\in \B_1^*$, to be the
element of $\B^*$ defined as $b \mapsto b_1^*(A(b))$.

A Banach space-valued random variable $X$ is a measurable function $X$
defined on a given probability space and taking values in $\B$
$$X:(\Omega, \cF) \rightarrow (\B, \cB),$$  
where $\cB$ denotes Borel $\sigma$-algebra of subsets of $\B$. The
separability assumption implies that $\cB$ is also the
$\sigma$-algebra generated by the continuous linear functionals on
$\B$, i.e., the smallest $\sigma$-algebra with respect to which all
elements of $\B^*$ are measurable.

If $X$ is a $\B$-valued random variable and $\E(\|X\|) < \infty$, then
for every $f\in\B^*$ the real-valued random variable $f(X)$ has finite
expectation, since $\E(|f(X)|) \leq \|f\| \E(\|X\|) < \infty$. Clearly
the functional $f\mapsto \E(f(X))$ is linear in $f\in\B^*$ and the
previous inequality shows that it is continuous at zero, hence
defining a continuous linear functional on $\B^*$, i.e., an element of
$\B^{**}$. It can be shown that this element of $\B^{**}$ has the form
$f\mapsto f(\mu)$ for a vector $\mu\in\B$. We call this vector the
expected value, or expectation, of $X$ and write $\mu = \E(X)$. The
expected value can be characterized as the unique element of $\B$ such
that $\E(f(X)) = f(\mu)$ $\forall f\in\B^*$. Expected values commute
with continuous linear operators in the following sense: if
$A\in\Lin(\B, \B_1)$ and $X$ is a $\B$-valued random variable with
expected value $\mu$, then the $B_1$-valued random variable $X_1 =
A(X)$ has expected value given by $\E(X_1) = A(\mu)$.

An argument along
similar lines can be used to show that if $\E(\|X\|^2) < \infty$, then
the mapping $\lambda: \B^*\times\B^*\longrightarrow \R$
specified by
$$\lambda(f, g) = \E\bigl( f(X-\mu) g(X-\mu)\bigr),\qquad f, g\in
\B^*$$ defines a bilinear function which, in turn, identifies a unique
continuous linear operator $\Lambda\in\Lin(\B^*, \B)$ via the identity
$\lambda(f, g) = f(\Lambda(g))$.  The operator $\Lambda$ is called the
covariance operator, or just covariance, of $X$, while the bilinear
function $\lambda$ is called the covariance function of
$X$. Covariance operator and covariance function are two equivalent
ways of providing the same information about the distribution of
$X$. A covariance operator is symmetric and positive, i.e.,
$g(\Lambda(f)) = f(\Lambda(g))$ and $f(\Lambda(f))\geq 0$ for all
$f,g\in \B^*$.  Unlike what happens in $\R^n$, where every symmetric,
positive definite matrix is a covariance matrix, not all symmetric and
positive elements of $\Lin(\B^*, \B)$ are valid covariance
operators. In fact, $\Lambda\in\Lin(\B^*, \B)$ is a covariance
operator if and only if there is a sequence $\{x_n\}$ in $\B$ with
$\sum_n \|x_n\|^2 <\infty$ such that $\Lambda(f) = \sum_n f(x_n) x_n$
for all $f\in\B^*$.

If $X_1$ and $X_2$ are $B_1$- and $B_2$-valued random variables,
respectively, with $\E(\|X_i\|^2)<\infty$, $i=1,2$, and expected
values $\mu_1$ and $\mu_2$, then we can define the covariance function
between $X_1$ and $X_2$ to be the bilinear operator
$\lambda_{12}:\B_1^*\times\B_2^*\longrightarrow \R$ defined by
$$\lambda_{12}(f_1, f_2) = \E\bigl( f_1(X_1-\mu_1)
f_2(X_2-\mu_2)\bigr),\qquad f_i\in \B_i^*, i=1, 2.$$ The corresponding
covariance operator $\Lambda_{1,2}\in\Lin(\B_2^*, \B_1)$ between $X_1$
and $X_2$ is determined by the relationship
$$f_1(\Lambda_{12}(f_2)) = \lambda_{12}(f_1, f_2),\qquad f_i\in
\B_i^*, i=1, 2.$$

A $\B$-valued random variable $X$ with expected value $\mu$ and
covariance $\Lambda$ has a Gaussian distribution if for every
$f\in \B^*$ the real-valued random variable $f(X)$ has a Gaussian
distribution. In this case, we write $X\sim\Nor_\B(\mu, \Lambda)$. It
is not hard to show that $X$ has a Gaussian distribution if and only
if its characteristic functional (Fourier transform) has the form
$$\psi(f) = \E\big(e^{if(X)}\big) = \exp\left\{if(\mu) -
  \tfrac{1}{2}\lambda(f, f)\right\},\qquad f\in\B^*,$$ where $\lambda$
is the covariance function associated with $\Lambda$. Using this
characterization, it is easy to see that, if $A\in \Lin(\B, \B_1)$,
then $A(X)\sim \Nor_{\B_1}(A(\mu), A\Lambda A^*)$.

Unlike what happens in the finite-dimensional case, for a general $\B$
there are valid covariance operators that are not the covariance
operator of any $\B$-valued Gaussian random variable.

To conclude this section, let us recall the definition of regular
conditional distribution. Let $Z$ be a random variable taking values
in a measurable space $(\bS, \mathcal{S})$, and let $\cG$ be a
sub-$\sigma$-algebra of $\cF$. A function $\pi: \Omega\times
\mathcal{S}: \longrightarrow \R$ is a regular conditional distribution
(r.c.d.) for $Z$ given $\cG$ if the following two conditions hold.
\begin{enumerate}
\item For every $\omega\in\Omega$, $\pi(\omega, \cdot)$ is a
  probability on $(\bS, \mathcal{S})$.
\item For every $S\in \mathcal{S}$, $\pi(\cdot, S)$ is a version of
  $\P(Z\in S \mid \cG)$.
\end{enumerate}
A standard result about r.c.d.'s is that if $\bS$ is a Polish space
with Borel $\sigma$-algebra $\mathcal{S}$, then a r.c.d. for $Z$ given
$\cG$ exists. In particular, this is the case when $\bS$ is a
separable Banach space endowed with its Borel $\sigma$-algebra.  For
notational simplicity, in the following sections we will typically
omit the explicit dependence on $\omega$ of a regular conditional
probability.

\section{Functional dynamic linear model}
\label{sec:fdlm} %
We define in this section the functional dynamic linear model (FDLM)
and we discuss the extension of Kalman filtering and smoothing
recursions, valid in the finite-dimensional case, to the case of
Banach space-valued states and observations. We assume that the reader
is familiar with the basic elements of dynamic linear models (DLMs)
from a Bayesian perspective in the standard case of finite-dimensional
states and observations, as found for example in
\citet{West+Harrison:1997} or \citet{Petris+Petrone+Campagnoli:2009}.

Let $\Fban$, the observation space, and $\Gban$, the state space, be
separable Banach spaces endowed with their Borel
$\sigma$-algebras. Consider infinite sequences $Y_1, Y_2, \dots$ and
$X_0, X_1, \dots$ of $\Fban$- and $\Gban$-valued random variables. We say
that they form a state space model if $\{X_t\}$ is a Markov
chain and, for every $t$, the conditional distribution of $Y_t$ given
all the other random variables depends on the value of $X_t$ only.
Let $F\in \Lin(\Gban, \Fban)$ and $G\in \Lin(\Gban)$. An FDLM is a
state space model satisfying the following distributional assumptions:
\begin{equation}
  \begin{aligned}
    X_0 &\sim\Nor_\Gban \big(m_0,C_0\big),\\[2pt]
    X_t|X_{t-1} = x_{t-1} &\sim \Nor_\Gban\big(G(x_{t-1}), W\big),\\[2pt]
    Y_t|X_t = x_t &\sim \Nor_\Fban\big(F(x_t), V\big),
  \end{aligned}
  \label{eq:fdlm}
\end{equation}
where $m_0\in \Gban$, $C_0$ and $W$ are covariance operators on
$\Gban$, and $V$ is a covariance operator on $\Fban$. The definition
as well as Kalman recursions, given below, can be extended in an
obvious way to time-dependent operators $F$, $G$, $V$ and $W$; we use
the time-invariant version of the model in this paper mainly for
notational simplicity. As for the finite dimensional DLM, quantities
of immediate interest related to this model are the filtering and
smoothing distributions, that is, the conditional distribution of the
state $X_t$ given the observations $Y_{1:t}$ (filtering distribution)
and the conditional distribution of $X_s$, for $s\leq t$, given
$Y_{1:t}$ (smoothing distribution). In the finite dimensional DLM all
the conditional distributions of a set of states or future
observations, given past observations, are Gaussian. This property
extends to the FDLM. One practical issue that arises in the infinite
dimensional model is that observations, while conceptually infinite
dimensional, have to be discretized at some point, in order to allow
proper data storage and processing. Clearly this discretization leads
in general to a loss of information. However, in this context, one
would hope that the inference based on the discretized data is almost
as good as the inference based on the complete, functional data, at
least if the discretized version of the data is still rich enough to
carry most of the information from the complete data. In other words,
one needs to show a continuity property of the inference -- the
filtering and smoothing distributions -- with respect to the
discretization. If the discretization is defined in a way that is
consistent with the infinite dimensional process, then one can show
that for the FDLM the continuity property mentioned above holds.

In order not to clutter the notation, we discuss the continuity of the
posterior distribution with respect to a sequence of discretizations
only in the case of one functional observation $Y$ and one functional
state $X$. Clearly, the argument extends to the FDLM setting in a
straightforward way.  Let $D_n\in \Lin (\Fban, \R^{d_n})$, $n\geq 1$,
be a sequence of linear, continuous operators. The $D_n$'s define by
composition a sequence of random variables $Y_n = D_n(Y)$, where $Y_n$
is $\R^{d_n}$-valued. We require that $\sigma(Y_n) \uparrow
\sigma(Y)$, which formally expresses the fact that the information
carried by the discretized version $Y_n$ approximates better and
better the information carried by the complete datum $Y$, coinciding
with it in the limit.  Let $\pi_n$ be a r.c.d of $X$ given $Y_n$ and
$\pi$ a r.c.d. of $X$ given $Y$. Then
$$\lim_{n\rightarrow\infty} \pi_n = \pi\qquad \text{almost surely},$$  
where the limit is in the topology of weak convergence of probability
measures. Moreover, since all the $\pi_n$ are Gaussian distributions,
and the class of Gaussian distributions is closed under the topology
of weak convergence, one can also deduce that $\pi$ is a Gaussian
distribution as well.

An example of a sequence of discretizations of the type described
above is the following. Consider $\Fban = C([0,1])$. 
For $n\geq 1$ and $1\leq k\leq 2^n$, let $q_{n,k} = k\, 2^{-n}$ and
define $D_n: C([0,1]) \longrightarrow \R^{2^n}$ by the formula
$$ D_n(y) = \big(y(q_{n,1}), \dots, y(q_{n, 2^n})\big).$$
Note that the same sequence of discretizing operators, evaluating a
function at the points of a sequence of finer and finer grids in
[0,1], would not be well defined if the functional datum $Y$ were an
element of $L^2([0,1])$, as it is often assumed in the FDA
literature. In fact, in that case the value of the function at any
given point is not even well defined, since elements of $L^2([0,1])$
are equivalence classes of functions, defined up to equality almost
everywhere.
 
Kalman filter and smoother, as well as the simulation smoother, or
forward filtering backward sampling algorithm (FFBS), which draws a
sample from the smoothing distribution, can be extended to the
FDLM. The following theorem provides the Kalman filter recursion for
the FDLM.

\begin{thm}
  \label{thm:Kalman}
  Consider the FDLM \eqref{eq:fdlm} and, for $n\geq 1$, let $D_n\in
  \Lin(\Fban, \R^{d_n})$. Define $Y_{t,n} = D_n(Y_t)$. Assume that
  $\sigma(Y_{t,n}) \uparrow \sigma(Y_t)$ and suppose that 
  $$X_{t-1}\mid Y_{1:t-1} \sim \Nor(m_{t-1}, C_{t-1}).$$
  Then the updating of the filtering distribution proceeds as follows.
  \begin{enumerate}
  \item One-step-ahead forecast distribution for the state:
    $$X_t\mid Y_{1:t-1} \sim\Nor(a_t, R_t),$$
    with $a_t = G(m_{t-1})$ and $R_t = GC_{t-1}G^* + W$.
  \item One-step-ahead forecast distribution for the observation:
    $$Y_t\mid Y_{1:t-1} \sim\Nor(f_t, Q_t),$$
    with $f_t = F(a_t)$ and $Q_t = FR_tF^* + V$.
  \item One-step-ahead forecast distribution for the discretized
    observation:
    $$Y_{t,n}\mid Y_{1:t-1} \sim\Nor(f_{t,n}, Q_{t,n}),$$
    with $f_{t,n} = D_n(f_t)$ and $Q_{t,n} = D_nQ_tD_n^*$.
  \item Filtering distribution at time $t$, given the discretized
    observation: 
    $$X_t\mid Y_{1:-1}, Y_{t,n} \sim\Nor(m_{t,n}, C_{t,n}),$$
    with $m_{t,n} = a_t + R_tF^*D_n^*Q_{t,n}^{-1} (Y_{t,n} - f_{t,n})$ and
    $C_{t,n} = R_t - R_tF^*D_n^*Q_{t,n}^{-1}D_nFR_t$.
  \item Filtering distribution at time $t$:
    $$\pi_t = \lim_{n\rightarrow\infty} \pi_{t,n}\qquad \text{a.s.},$$
    where $\pi_{t,n}$ is a r.c.d. for $X_t$ given $(Y_{1:t-1},
    Y_{t,n})$ and $\pi_t$ is a r.c.d. for $X_t$ given
    $Y_{1:t}$. Moreover, $\pi_t$ is a.s. a Gaussian distribution.
  \end{enumerate}
\end{thm}

\begin{proof}[Proof of Theorem~\ref{thm:Kalman}]
  We will first derive the joint conditional distribution of $(X_t,
  Y_t)$ given $Y_{1:t-1}$, from which the conditional distributions in
  1 and 2 will easily follow. The dual of $\Gban\times \Fban$ can be
  identified with $\Gban^*\times \Fban^*$ noting that the element
  $(x^*, y^*) \in \Gban^*\times \Fban^*$ can be associated to the
  element of $(\Gban\times \Fban)^*$ defined by
  $$\widetilde{(x^*, y^*)}: (x, y) \mapsto x^*(x) + y^*(y).$$
  Moreover, every element in $(\Gban\times \Fban)^*$ has that form for
  a unique choice of $x^*$ and $y^*$. We will make use of the
  following matrix notation for operators. If $A\in \Lin(\A, \A_1)$,
  $B\in \Lin(\B, \A_1)$, $C\in \Lin(\A, \B_1)$, and $D\in \Lin(\B,
  \B_1)$, the matrix
  $$
  \begin{bmatrix}
    A & B\\ C & D
  \end{bmatrix}
  $$
  denotes the element of $\Lin(\A\times\B, \A_1\times\B_1)$ defined by
  $$\A\times\B \ni (a, b) \mapsto \big(A(a) + B(b), C(a) +
  D(b)\big).$$ 
  It is easy to show that any element of $\Lin(\A\times\B,
  \A_1\times\B_1)$ can be uniquely represented in the matrix form
  written above.
  Let us compute the conditional characteristic functional of $(X_t,
  Y_t)$ given $Y_{1:t-1}$. For $x^*\in \Gban^*$ and $y^*\in \Fban^*$
  we have
  \begin{align*}
    \psi(x^*, y^*) &= \E\Big(\exp\bigg\{i\big( x^*(X_t) +
    y^*(Y_t)\big) \bigg\}| Y_{1:t-1}\Big)\\
    &= \E\Big(\E\Big(\exp\bigg\{i\big( x^*(X_t) + y^*(Y_t)\big)
    \bigg\}| X_t, Y_{1:t-1}\Big) | Y_{1:t-1}\Big) \\
    &= \E\Big(\exp\bigg\{i x^*(X_t) + i y^*(F(X_t)) - \frac{1}{2}
    y^*(V(y^*)) \bigg\}| Y_{1:t-1}\Big)\\
    &= \E\Big(\exp\bigg\{i \big(x^* + F^*y^*\big) (X_t) \bigg\}|
    Y_{1:t-1}\Big) \exp \bigg\{ -\frac{1}{2} y^*(V(y^*)) \bigg\}\\
    &= \E\Big(\E\Big(\exp\bigg\{i \big(x^* + F^*y^*\big) (X_t)
    \bigg\}| X_{t-1}, Y_{1:t-1}\Big)| Y_{1:t-1}\Big) \exp \bigg\{
    -\frac{1}{2} y^*(V(y^*)) \bigg\}\\
    &= \E\Big(\exp\bigg\{i \big(x^* + F^*y^*\big) (G(X_{t-1}))
    \bigg\}|
    Y_{1:t-1}\Big)\\
    &\qquad \exp \bigg\{ -\frac{1}{2} \Big[ (x^* + F^*y^*)\big(W(x^* +
    F^*y^*)\big) +
    y^*(V(y^*)) \Big] \bigg\}\\
    &= \exp\bigg\{i \big(x^* + F^*y^*\big) (G(m_{t-1}))\\
    &\qquad -\frac{1}{2}\Big[ (x^*G + F^*y^*G)\big( C_{t-1}(x^*G +
    F^*y^*G) \big)\\
    &\qquad + (x^* + F^*y^*)\big(W(x^* + F^*y^*)\big) +
    y^*(V(y^*)) \Big] \bigg\}\\
    &= \exp\bigg\{i \big[ x^*(G(m_{t-1})) + y^*(FG(m_{t-1})) \big]\\
    &\qquad -\frac{1}{2}\Big[ x^*\big((GC_{t-1}G^* + W)(x^*)\big)\\
    &\qquad + x^*\big((GC_{t-1}G^*F^* + WF^*)(y^*)\big)
    + y^*\big((FGC_{t-1}G^* + FW)(x^*)\big)\\
    &\qquad + y^*\big((FGC_{t-1}G^*F^* + FWF^* + V)(y^*)\big) \Big]\bigg\}
  \end{align*}
  This shows that the conditional distribution of $(X_t, Y_t)$ given
  $Y_{1:t-1}$ is Gaussian with mean $(a_t, f_t)$ and covariance
  operator 
  $$
  \begin{bmatrix}
    R_t & R_tF^*\\ FR_t & Q_t
  \end{bmatrix},
  $$
  from which parts 1 and 2 of the theorem follow. Part 3 is an
  immediate consequence of part 2, when one considers how Gaussian
  distributions transform under the application of a continuous linear
  operator. As far as part 4 is concerned, if $X_t$ were a finite
  dimensional random variable, then the result would be a
  straightforward application of the well-known theorem on Normal
  correlation \citep{Lipster+Shiryayev:1972, Barra:1981}. It is simple
  to verify that the proof of that result carries over to the case
  where $X_t$ is a Banach space-valued random variable, as long as the
  conditioning random variable is finite dimensional. Finally, part 5
  of the theorem follows from the result on discretization of
  observations discussed above.
\end{proof}

As far as the smoothing distribution is concerned, since under our
modelling assumptions the joint distribution of $(X_{0:t}, Y_{1:t})$
is Gaussian, all its marginal distributions, including that of $(X_s,
Y_{1:t})$ are Gaussian as well. It follows that the conditional
distribution of $X_s$ given $Y_{1:t}$ is again Gaussian; that is, the
smoothing distribution of $X_s$ is Gaussian, as in the finite
dimensional setting. In general, however, the smoothing means and
covariances do not have a simple explicit form. For the purpose of
applications this is not a big impediment, since the analysis is
always performed on a discretized version of the data, to which the
usual smoothing recurrence applies
\citep{Petris+Petrone+Campagnoli:2009}. The inference obtained from
the discretized version of the FDLM converges, as the discretization
gets finer, to the inference that one would obtain from the complete
functional data, by the argument discussed before
Theorem~\ref{thm:Kalman}.


\section{Application to $C([0,1])$-valued time series}
\label{sec:example} %
In order to define a Gaussian $\B$-valued random variable one can
rely, when $\B$ is a Banach space of functions, on the theory of
stochastic processes. When $\B = C([0,1])$, a stochastic process with
continous sample paths can be interpreted as a $\B$-valued random
variable. Let us spell out the equivalence, which will be used in the
rest of the present section. Suppose $\zeta =
\{\zeta_t: t\in [0,1]\}$ is a stochastic process with continous
trajectories. Note that every $\zeta_t$ is a random variable, i.e.,
$\zeta_t = \zeta_t(\omega)$. Then, since the sample paths are
continuous, one can define the function $\tilde\zeta: \Omega
\longrightarrow C([0,1])$ by setting 
\begin{equation} 
  \tilde\zeta(\omega): t \mapsto \zeta_t(\omega),\qquad 
  \text{$t\in [0,1]$ and $\omega\in \Omega$.}
  \label{eq:process}
\end{equation}
The following theorem shows that $\tilde\zeta$ is a $C([0,1])$-valued
random variable and specifies its mean and covariance function. In
addition, it shows that $\tilde\zeta$ has a Gaussian distribution if
the process $\zeta$ does. Recall that, by Riesz representation
theorem, $C([0,1])^*$ can be identified with the Banach space of all
signed measures on the Borel sets of [0,1], denoted below by
$\mathcal{M}([0,1])$. For $\eta\in \mathcal{M}([0,1])$ and $x\in
C([0,1])$ we will use the notation
$$\eta(x) := \int_{[0,1]} x(t)\eta(dt).$$ 

\begin{thm}
  \label{thm:process}
  For the function defined in \eqref{eq:process}, the following hold.
  \begin{enumerate}
  \item $\tilde\zeta$ is a measurable function from
    $(\Omega, \cF)$ to the Banach space $C([0,1])$ endowed wih its Borel
    $\sigma$-algebra.
  \item If, in addition, the process $\zeta$ possesses second moments,
    then the expected value and the covariance function of
    $\tilde\zeta$ are given by
    \begin{gather*}
      C([0,1])\ni\E(\tilde\zeta): t\mapsto \E(\zeta_t),\\[5pt]
      \lambda(\eta, \tau) = \int_{[0,1]^2} \gamma(u,v)
      \eta(du)\tau(dv),\qquad \eta, \tau \in\mathcal{M}([0,1]),
    \end{gather*}
    where $\gamma(u,v) = \Cov(\zeta_u, \zeta_v)$. The covariance
    operator of $\tilde\zeta$ is
    \begin{equation}
      \label{eq:covop_OU}
      \Lambda(\eta) = \int_{[0,1]} \gamma(u, \cdot)\, \eta(du), \qquad
      \eta\in \mathcal{M}([0,1]).
    \end{equation}
  \item If, in addition, the process $\zeta$ is Gaussian, then
    $\tilde\zeta$ has a Gaussian distribution.
  \end{enumerate}
\end{thm}
\begin{proof}
  \begin{enumerate}
  \item See \citet{Bosq:2000}, Example 1.10.
  \item Let $\eta, \tau\in\mathcal{M}([0,1])$. In view of Jordan
    decomposition $\eta = \eta^+ - \eta^-$, so we can assume, without
    real loss of generality, that $\eta$ and $\tau$ are positive measures.
    By a straightforward application of Fubini's theorem, we have
    \begin{align*}
      \E\bigg(\int_{[0,1]} \tilde\zeta(t) \eta(dt)\bigg) &=
      \int_\Omega\int_{[0,1]} \zeta_t(\omega) \eta(dt) P(d\omega)\\
      &= \int_{[0,1]}\int_\Omega \zeta_t(\omega) P(d\omega) \eta(dt)\\
      &= \int_{[0,1]} \E(\zeta_t) \eta(dt) = \eta\big(\E(\zeta_{\cdot})\big).
    \end{align*}
    Let $m = \E(\tilde\zeta)$. Then, using Fubini's theorem,
    \begin{align*}
      \lambda(\eta, \tau) &= \E\big(\eta(\tilde\zeta - m)
      \tau(\tilde\zeta - m) \big)\\
      &= \int_\Omega \Big(\int_{[01]} \big(\zeta_u(\omega) -
      m(u)\big) \eta(du)\Big) \Big(\int_{[0,1]} \big(\zeta_v(\omega)
      - m(v)\big) \tau(dv)\Big) \P(d\omega)\\
      &= \int_\Omega \Big(\int_{[01]^2} \big(\zeta_u(\omega) -
      m(u)\big) \big(\zeta_v(\omega) - m(v)\big) \eta(du)\tau(dv)
      \Big) \P(d\omega)\\
      &= \int_{[01]^2}\Big( \int_\Omega \big(\zeta_u(\omega) -
      m(u)\big) \big(\zeta_v(\omega) - m(v)\big) \P(d\omega) \Big)
      \eta(du)\tau(dv)\\
      &= \int_{[01]^2} \gamma(u,v) \eta(du)\tau(dv) . 
    \end{align*}
    The form of the covariance operator follows immediately from the
    expression giving the covariance function.
  \item Let $\eta\in\mathcal{M}([0,1])$ and consider the
    discretization operator defined for any $x\in C([0,1])$ by
    $$x_n(t) = \I_{\{0\}}(t) x(0) + \sum_{j=1}^n \I_{\Delta_{n,j}}(t)
    x(j/n),$$
    where 
    $$\Delta_{n,j} = \left(\frac{j-1}{n}, \frac{j}{n}\right].$$
    It is clear that the map $x\mapsto x_n$ is non-expansive, i.e.,
    $\| y_n - x_n \| \leq \|y - x\|$, and therefore continuous. It
    follows that by applying this operator to $\tilde\zeta$ we obtain
    another $C([0,1])$-valued random variable, say
    $\tilde\zeta_n$. For any fixed $\omega\in\Omega$,
    $$\eta\big(\tilde\zeta(\omega)\big) = \eta(\{0\}) \zeta_0(\omega)
    + \sum_{j=1}^n \eta(\Delta_{n,j}) \zeta_{j/n}(\omega).$$ Since the
    stochastic process $\zeta$ is Gaussian, the joint distribution of
    $(\zeta_0, \zeta_{1/n},\ldots, \zeta_1)$ is Gaussian, hence
    $\eta(\tilde\zeta)$ is Gaussian as well.  It is also easy to show
    that, with probability one, $\tilde\zeta_n\rightarrow \tilde\zeta$
    as $n\rightarrow\infty$. Since $\eta$, as a functional on
    $C([0,1])$, is continuous, it follows that
    $\eta(\tilde\zeta_n)\rightarrow \eta(\tilde\zeta)$ almost surely
    and, a fortiori, $\eta(\tilde\zeta_n)\stackrel{d}{\rightarrow}
    \eta(\tilde\zeta)$. Since the class of Gaussian distributions on
    $\R$ is closed with respect to the topology of weak convergence,
    we conclude that $\eta(\tilde\zeta)$ has a Gaussian distribution.
  \end{enumerate}
\end{proof}

In our numerical example below we will make extensive use of
Theorem~\ref{thm:process}, using it to define Gaussian
$C([0,1])$-valued random variables starting from the
Ornstein-Uhlenbeck process, having mean zero and covariance function
\begin{equation}
  \label{eq:covfun_OU}
  \gamma(u, v) = \frac{\sigma^2}{2\beta} \exp \bigl\{-\beta |u-v|
  \bigr\},
\end{equation}
where $\sigma^2$ and $\beta$ are positive parameters. These
random variables will be used, in turn, as building blocks to set up
an FDLM.

We consider a data set consisting in hourly measurements on the log
scale of electricity demand, over the previous hour, collected at
a distribution station in the Northeastern region of the United States
from January 2006 to December 2010. We consider the data to be a
discretized version of a daily functional time series. Since it is
reasonable to assume that electricity demand follows a continuous
path over time, we will model the data as $C([0,1])$-valued random
variables. Figure~\ref{fig:dataFull} shows the full data set.
\begin{figure}
  \centering
  \includegraphics[draft=false]{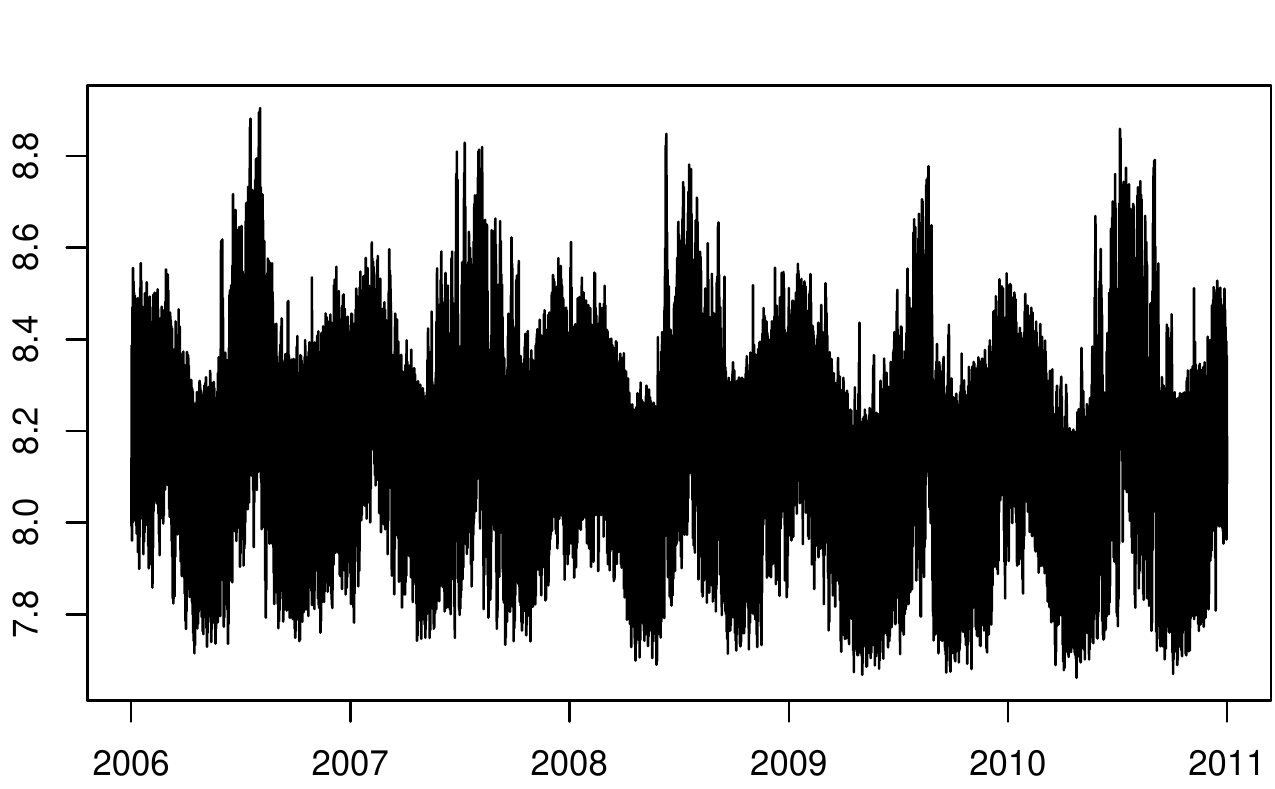}
  \caption{Electricity demand}
  \label{fig:dataFull}
\end{figure}

For scalar time series, a specific DLM that has been successfully used
to model observations with a constant or slowly changing mean is the
so-called local level model. This simply consists in a random walk for
a univariate state, which is observed with noise. The model can be
immediately extended to functional data, setting $\Fban = \Gban$, $F =
G = \1_{\Gban}$ in \eqref{eq:fdlm}, where, for any Banach space $\B$,
$\1_\B$ denotes the identity operator on $\B$. We take $m_0$ to be the
zero element of $C([0,1])$, and the covariance operators $C_0$, $W$
and $V$ to be of the form~\eqref{eq:covop_OU}, with $\gamma(u, v)$
specified in~\eqref{eq:covfun_OU}. The parameters $\sigma^2$ and
$\beta$ in ~\eqref{eq:covfun_OU} are different for the three
covariance operators and, while we fix their value when we define
$C_0$, so as to obtain a prior distribution for the initial state that
is only vaguely informative, we estimate the parameters of $V$ and
$W$, $(\sigma^2_V, \beta_V)$ and $(\sigma^2_W, \beta_W)$,
respectively. The inference was carried out using MCMC, simulating in
turn the latent states via the forward filtering backward sampling
(FFBS) algorithm \citep{Carter+Kohn:1994, Fruewirth-Schnatter:1994,
  Shephard:1994}, and the parameters $\sigma^2_V, \beta_V, \sigma^2_W,
\beta_W$. We coded the sampler in the statistical programming language
R \citep{R}, using also the contributed packages \texttt{zoo}
\citep{Zeileis+Grothendieck:2005} for data manipulation and graphing,
and \texttt{dlm} \citep{DLM:2010}, which contains an implementation of
FFBS.
\begin{figure}
  \centering
  \includegraphics[draft=false]{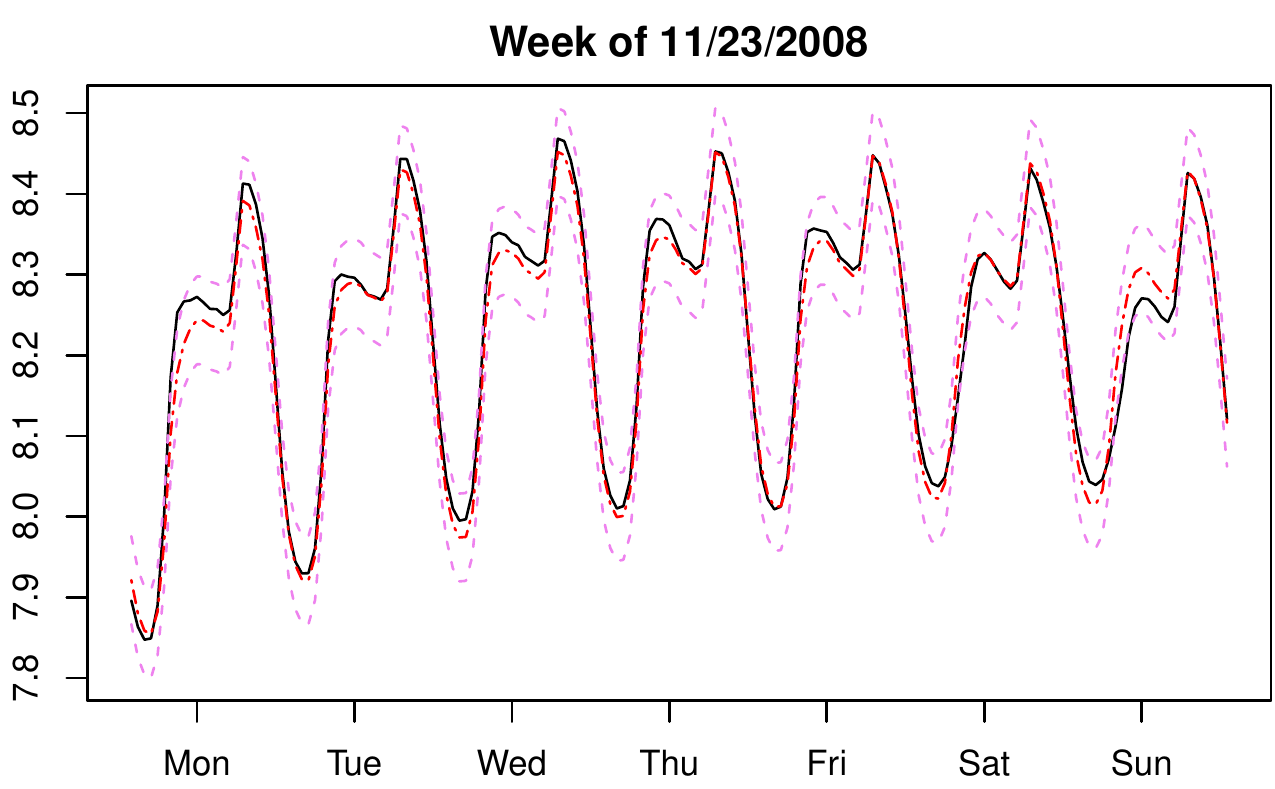}
  \caption{One week of electricity demand (solid line), with smoothed
    demand (two dash) and 90\% probability bands (dashed)}
  \label{fig:smooth_week}
\end{figure}
The prior used for the
two variance parameters is an inverse gamma, which is 
conditionally conjugate for this particular model, when the latent
states are included in the simulation, while the two remaining
parameters $\beta_V$ and $\beta_W$ were updated with a random walk
Metropolis-Hastings step on the log
scale. Figure~\ref{fig:smooth_week} displays, for one particular week,
the observations together with the smoothed states for those seven
days, and 90\% probability bands, obtained from the MCMC output. The
fit is good, showing that the functional model, despite the small
number of parameters, is flexible enough to adapt and learn the
general daily pattern of electricity demand on any given day. 

\begin{table}
  \centering
  \begin{tabular}{cccc}
    \hline\hline
    $\sigma^2_V$ & $\log \beta_V$ & $\sigma^2_W$ & $\log \beta_W$\\
    \hline
       $2.76\cdot 10^{-04}$ &  $-2.83$ &  $2.14\cdot 10^{-04}$ &  $-3.23$\\
       $9.86\cdot 10^{-08}$ & $2.30\cdot 10^{-03}$ &  $1.33\cdot
       10^{-07}$ & $2.51\cdot 10^{-03}$\\
       $(2.70, 2.81)\cdot 10^{-4}$ & $(-2.89, -2.76)$ & $(2.09,
       2.20)\cdot 10^{-4}$ & $(-3.30, -3.16)$\\
       \hline\hline
  \end{tabular}
  \caption{Posterior estimates, with Monte Carlo standard errors and
    90\% posterior probability intervals}
  \label{tab:par_est}
\end{table}
In terms of the inference on the model parameters,
Table~\ref{tab:par_est} summarizes posterior estimates of the four
parameters, together with MC standards errors and 90\% posterior
probability intervals. MC standard errors are computed using Sokal's
estimator \citep{Sokal:1989}, as implemented in the R package
\texttt{dlm}.

\section{Conclusions}
\label{sec:conclusion} %
The model presented in the paper is an important step forward in the
methodology of analysis of functional time series. For such kind of
data it provides a much more flexible setting compared to functional
ARMA models \citep{Bosq:2000, Horvath+Kokoszka:2012}. The FDLM allows
to extend to the functional setting most of the standard structural
time series models \citep{Harvey:1989} that have proved extremely
useful for the analysis and forecasting of finite dimensional time
series. Among the advantages of the FDLM proposed in the paper, we
note that the specification of a particular model is in most cases
relatively straightforward, as illustrated in
Section~\ref{sec:example}, and the practical implementation of the
posterior sampling can be done using standard MCMC algorithms.

\clearpage
\bibliography{fdlm}

\end{document}